\documentclass[a4paper,12pt]{article}
\usepackage{graphicx}
\usepackage{amssymb,amsmath,amsfonts ,amsthm}
\usepackage{ifthen,verbatim,epsfig,color}
\usepackage[mathcal]{euscript}
\usepackage{hyperref}
\usepackage[capitalize]{cleveref}
\usepackage{complexity}

\hypersetup{colorlinks=true,linkcolor=blue,citecolor=brown,pdfborderstyle={/S/U/W 1}}
\usepackage{subcaption}

\usepackage{doi}
\usepackage{stmaryrd}    
\usepackage{todonotes}
\usepackage{xspace}
\usepackage{authblk}

\usepackage{fullpage}

\theoremstyle{plain}
\newtheorem{theorem}{Theorem} 
\newtheorem{lemma}  [theorem]   {Lemma}
\newtheorem{corollary}[theorem] {Corollary}
\newtheorem{fact} [theorem] {Fact}
\newtheorem{proposition} [theorem] {Proposition}

\newtheorem{conjecture}[theorem] {Conjecture}

\theoremstyle{definition}

\newtheorem{definition}[theorem] {Definition}

\DeclareMathOperator{\Col}{Col}
\DeclareMathOperator{\CSP}{CSP}
\DeclareMathOperator{\Hom}{Hom}

\DeclareMathOperator{\Mix}{Mix}

\DeclareMathOperator{\Recol}{Recol}
\DeclareMathOperator{\Recon}{Recon}
\DeclareMathOperator{\alt}{alt}

\makeatletter
\newcommand{\mylabel}[2]{#2\def\@currentlabel{#2}\label{#1}}
\makeatother

 \newcommand{\cE}{\mathcal{E}}
 \newcommand{\cK}{\mathcal{K}}

 \newcommand{\B}{\cK}  
 \newcommand{\Cl}{K} 
 \newcommand{\rCol}{\Col}

\title{Reconfiguring homomorphisms to reflexive graphs via a simple reduction}

\author[1]{Moritz M\"uhlenthaler}
\author[2]{Mark H. Siggers}
\author[1]{Thomas Suzan\footnote{Corresponding author}}

\affil[1]{
Laboratoire G-SCOP, Grenoble INP, Universit\'e Grenoble-Alpes, France\\
            \texttt{\{moritz.muhlenthaler,thomas.suzan\}@grenoble-inp.fr}}
\affil[2]{Kyungpook National University, S. Korea\\
    \texttt{mhsiggers@knu.ac.kr}}

\date{}

\usetikzlibrary{shapes,backgrounds,plothandlers,plotmarks,calc,arrows,fadings,patterns,decorations.pathmorphing}

\pgfdeclarepatternformonly[\LineSpace,\LineWidth]{my horizontal lines}%
    {\pgfpointorigin}{\pgfqpoint{100pt}{1pt}}{\pgfqpoint{100pt}{\LineSpace}}%
    {
        \pgfsetcolor{\tikz@pattern@color}
        \pgfsetlinewidth{\LineWidth}
        \pgfpathmoveto{\pgfqpoint{0pt}{0.5pt}}
        \pgfpathlineto{\pgfqpoint{100pt}{0.5pt}}
        \pgfusepath{stroke}
    }
 
\pgfdeclarepatternformonly[\LineSpace,\LineWidth]{my vertical lines}%
    {\pgfpointorigin}{\pgfqpoint{1pt}{100pt}}{\pgfqpoint{\LineSpace}{100pt}}%
    {
        \pgfsetcolor{\tikz@pattern@color}
        \pgfsetlinewidth{\LineWidth}
        \pgfpathmoveto{\pgfqpoint{0.5pt}{0pt}}
        \pgfpathlineto{\pgfqpoint{0.5pt}{100pt}}
        \pgfusepath{stroke}
    } 
 
\pgfdeclarepatternformonly[\LineSpace,\LineWidth]{my grid}%
    {\pgfqpoint{-1pt}{-1pt}}{\pgfqpoint{\LineSpace}{\LineSpace}}
    {\pgfqpoint{\LineSpace}{\LineSpace}}%
    {
        \pgfsetcolor{\tikz@pattern@color}
        \pgfsetlinewidth{\LineWidth}
        \pgfpathmoveto{\pgfqpoint{0pt}{0pt}}
        \pgfpathlineto{\pgfqpoint{0pt}{\LineSpace + 0.1pt}}
        \pgfpathmoveto{\pgfqpoint{0pt}{0pt}}
        \pgfpathlineto{\pgfqpoint{\LineSpace + 0.1pt}{0pt}}
        \pgfusepath{stroke}
    }
 
\pgfdeclarepatternformonly[\LineSpace,\LineWidth]{my north east lines}
    {\pgfqpoint{-\LineWidth}{-\LineWidth}}{\pgfqpoint{\LineSpace}{\LineSpace}}
    {\pgfqpoint{\LineSpace}{\LineSpace}}%
    {
        \pgfsetcolor{\tikz@pattern@color}
        \pgfsetlinewidth{\LineWidth}
        \pgfpathmoveto{\pgfqpoint{-\LineWidth}{-\LineWidth}}
        \pgfpathlineto{\pgfqpoint{\LineSpace + 0.1pt}{\LineSpace + 0.1pt}}
        \pgfusepath{stroke}
    }
 
\pgfdeclarepatternformonly[\LineSpace,\LineWidth]{my north west lines}
    {\pgfqpoint{-\LineWidth}{-\LineWidth}}{\pgfqpoint{\LineSpace}{\LineSpace}}
    {\pgfqpoint{\LineSpace}{\LineSpace}}%
    {
        \pgfsetcolor{\tikz@pattern@color}
        \pgfsetlinewidth{\LineWidth}
        \pgfpathmoveto{\pgfqpoint{-\LineWidth}{\LineSpace}}
        \pgfpathlineto{\pgfqpoint{\LineSpace + 0.1pt}{-\LineWidth}}
        \pgfusepath{stroke}
    }
 
\pgfdeclarepatternformonly[\LineSpace,\LineWidth]{my crosshatch}%
    {\pgfqpoint{-1pt}{-1pt}}{\pgfqpoint{\LineSpace}{\LineSpace}}
    {\pgfqpoint{\LineSpace}{\LineSpace}}%
    {
        \pgfsetcolor{\tikz@pattern@color}
        \pgfsetlinewidth{\LineWidth}
        \pgfpathmoveto{\pgfqpoint{\LineSpace + 0.1pt}{0pt}}
        \pgfpathlineto{\pgfqpoint{0pt}{\LineSpace + 0.1pt}}
        \pgfpathmoveto{\pgfqpoint{0pt}{0pt}}
        \pgfpathlineto{\pgfqpoint{\LineSpace + 0.1pt}{\LineSpace + 0.1pt}}
        \pgfusepath{stroke}
    }
\pgfdeclarepatternformonly[\LineSpace,\LineWidth]{my dots}%
    {\pgfpointorigin}
    {\pgfpoint{\LineSpace}{\LineSpace*0.866025*2}}
    {\pgfpoint{\LineSpace}{\LineSpace*0.866025*2}}%
    {
        \pgfsetcolor{\tikz@pattern@color}
        \pgfpathcircle{\pgfpoint{\LineSpace*0.25}{\LineSpace*0.866025*0.5}}{\LineWidth}
        \pgfpathcircle{\pgfpoint{\LineSpace*0.75}{\LineSpace*0.866025*1.5}}{\LineWidth}        
        \pgfusepath{fill}
    }
\makeatother
 
\newdimen\LineSpace
\newdimen\LineWidth
\tikzset{
    pattern space/.code={\LineSpace=#1},
    pattern space=3pt,
    pattern width/.code={\LineWidth=#1},
    pattern width=.4pt
}

\definecolor{cdred}{RGB}{200,20,20}
\definecolor{cdgreen}{RGB}{220,255,200}	
\definecolor{cdblue}{RGB}{130,145,255}	
\tikzset{ %
	dred/.style={pattern=my north west lines,pattern width=.6pt,pattern space=2.3pt,pattern color=red},
	dgreen/.style={pattern=my grid,pattern width=0.8pt,pattern space=2pt,pattern color=green!80!black},
	dblue/.style={pattern=my dots,pattern width=0.7pt,pattern space=2.3pt,pattern color=blue},	
	H/.style={circle,fill=gray!20,draw=gray!20,line width=2pt,inner sep=0pt,minimum size=15pt},
	He/.style={draw=gray!20,line width=12pt},
	He2/.style={draw=gray!20,line width=6pt},
	G/.style={circle,fill=black,inner sep=0pt,minimum size=3pt},
	G1/.style={G,minimum size=5pt},
	Ge/.style={draw=black},
	v/.style={circle,draw=black!75,inner sep=0pt,minimum size=8pt},
	v1/.style={v,line width=1.5pt},
	Gv/.style={inner sep=0pt},
	Gve/.style={Ge,dashed},
	circ/.style={circle, fill=black, inner sep=2pt, node contents={}}
}

\usepackage{tikz}
\usepackage{tkz-graph}
\usepackage{tkz-berge}
\usetikzlibrary{matrix,calc,fit,shapes,patterns,backgrounds,decorations.pathmorphing,positioning,snakes,graphs,graphs.standard}

\usetikzlibrary{shapes,backgrounds,plothandlers,plotmarks,calc,arrows,fadings,patterns,decorations.pathmorphing}

\tikzset{
	invisible/.style={opacity=0},
	visible on/.style={alt={#1{}{invisible}}},
	alt/.code args={<#1>#2#3}{%
	  \alt<#1>{\pgfkeysalso{#2}}{\pgfkeysalso{#3}} 
	},
        cross/.style={cross out, thick,draw=black, minimum size=2*(#1-\pgflinewidth), inner sep=0pt, outer sep=0pt},
        cross/.default={0.25em},
        edge/.style={thick,black},
	arc/.style={thick,black,->},
        medge/.style={decorate,very thick,decoration={snake}},
        aedge/.style={very thick,dashed,black},
        dedge/.style={thick,->},
        availedge/.style={thick,blue},
        vertex/.style={inner sep=0.25em,shape=circle,thick,draw,node distance=2em},
	smalledge/.style={thick,DodgerBlue},
        smallvertex/.style={inner sep=0.1em,shape=circle,thick,draw=DodgerBlue,node distance=2em}
}

\tikzset{
    pattern space/.code={\LineSpace=#1},
    pattern space=3pt,
    pattern width/.code={\LineWidth=#1},
    pattern width=.4pt
}

\definecolor{cdred}{RGB}{200,20,20}
\definecolor{cdgreen}{RGB}{220,255,200}	
\definecolor{cdblue}{RGB}{130,145,255}	
\tikzset{ %
	dred/.style={pattern=my north west lines,pattern width=.6pt,pattern space=2.3pt,pattern color=red},
	dgreen/.style={pattern=my grid,pattern width=0.8pt,pattern space=2pt,pattern color=green!80!black},
	dblue/.style={pattern=my dots,pattern width=0.7pt,pattern space=2.3pt,pattern color=blue},	
	H/.style={circle,fill=gray!20,draw=gray!20,line width=2pt,inner sep=0pt,minimum size=15pt},
	H2/.style={circle,fill=gray!20,draw=gray!20,line width=2pt,inner sep=0pt,minimum size=30pt},
	He/.style={draw=gray!20,line width=12pt},
	He2/.style={draw=gray!20,line width=6pt},
	G/.style={circle,fill=black,inner sep=0pt,minimum size=3pt},
	G1/.style={G,minimum size=5pt},
	Ge/.style={draw=black},
	Ge2/.style={draw=black, line width=1.5pt},
	v/.style={circle,draw=black!75,inner sep=0pt,minimum size=8pt},
	v1/.style={v,line width=1.5pt},
	Gv/.style={inner sep=0pt},
	Gve/.style={Ge,dashed},
}

\newcommand{\dsevengraph}{
	\begin{scope}
		\node[H, label =  {[red!40,label distance=2mm]90:0}] (h0) at (0,0){};
		\node[H, label =  {[red!40,label distance=2mm]-120:2}] (ha2) at (-2,-1){};
		\node[H, label =  {[red!40,label distance=2mm]120:1}] (ha1) at (-2,1){};
		\node[H, label =  {[red!40,label distance=2mm]-60:3}] (hb1) at (2,-1){};
		\node[H, label =  {[red!40,label distance=2mm]60:4}] (hb2) at (2,1){};
		\node[H, label =  {[red!40,label distance=2mm]-120:5}] (hc1) at (-1,-3){};
		\node[H, label =  {[red!40,label distance=2mm]-60:6}] (hc2) at (1,-3){};
		
		\draw [blue!20] (h0.center) -- (ha1.center) -- (ha2.center) -- cycle;
		\draw [blue!20] (h0.center) -- (hb1.center) -- (hb2.center) -- cycle;

		\draw[He2, color = blue!20] (h0.center) --(ha1.center) -- (ha2.center) -- (h0.center);
		\draw[He2, color = blue!20] (h0.center) --(hb1.center) -- (hb2.center) -- (h0.center);
		\draw[He2, color = red!20] (ha2.center) -- (hc1.center) -- (hc2.center) -- (hb1.center);
		
		\node[H] (h0) at (0,0){};
		\node[H] (ha2) at (-2,-1){};
		\node[H] (ha1) at (-2,1){};
		\node[H] (hb1) at (2,-1){};
		\node[H] (hb2) at (2,1){};
		\node[H] (hc1) at (-1,-3){};
		\node[H] (hc2) at (1,-3){};
	\end{scope}
}

\newcommand{\dBsevengraph}{
	\begin{scope}
		\node[H, label =  {[red!40,label distance=2mm]180:0}] (h0) at (-1,3){};
		\node[H, label =  {[red!40,label distance=2mm]180:1}] (h1) at (-1,2){};
		\node[H, label =  {[red!40,label distance=2mm]180:2}] (h2) at (-1,1){};
		\node[H, label =  {[red!40,label distance=2mm]180:3}] (h3) at (-1,0){};
		\node[H, label =  {[red!40,label distance=2mm]180:4}] (h4) at (-1,-1){};
		\node[H, label =  {[red!40,label distance=2mm]180:5}] (h5) at (-1,-2){};
		\node[H, label =  {[red!40,label distance=2mm]180:6}] (h6) at (-1,-3){};
	
		\node[H,fill=blue!20, label =  {[red!20,label distance=2mm]0:012}] (hT1) at (1,2){};
		\node[H,fill=blue!20, label =  {[red!20,label distance=2mm]0:034}] (hT2) at (1,1){};
	
		\node[H,fill=red!20,label =  {[red!40,label distance=2mm]0:25}] (hE1) at (1,0){};
		\node[H,fill=red!20,label =  {[red!40,label distance=2mm]0:56}] (hE2) at (1,-1){};
		\node[H,fill=red!20,label =  {[red!40,label distance=2mm]0:36}] (hE3) at (1,-2){};
		
		\draw[He2] (hT1) -- (h0);
		\draw[He2] (hT1) -- (h1);
		\draw[He2] (hT1) -- (h2);
		\draw[He2] (hT2) -- (h0);
		\draw[He2] (hT2) -- (h3);
		\draw[He2] (hT2) -- (h4);
		\draw[He2] (hE1) -- (h2);
		\draw[He2] (hE1) -- (h5);
		\draw[He2] (hE2) -- (h5);
		\draw[He2] (hE2) -- (h6);
		\draw[He2] (hE3) -- (h6);
		\draw[He2] (hE3) -- (h3);
	\end{scope}
}

\begin{document}

\maketitle

\begin{abstract}
    Given a graph $G$ and two graph homomorphisms $\alpha$ and $\beta$ from $G$ to a fixed graph $H$, the problem $H$-Recoloring asks whether there is a transformation from $\alpha$ to $\beta$ that changes the image of a single vertex at each step and keeps a graph homomorphism throughout. The complexity of the problem depends among other things on the presence of loops on the vertices. We provide a simple reduction that, using a known algorithmic result for $H$-Recoloring for square-free irreflexive graphs $H$, yields a polynomial-time algorithm for $H$-Recoloring for square-free reflexive graphs $H$. This generalizes all known algorithmic results for $H$-Recoloring for reflexive graphs $H$. Furthermore, the construction allows us to recover some of the known hardness results. Finally, we provide a partial inverse of the construction for bipartite instances.
\end{abstract}

\noindent \textbf{Keywords}: graph homomorphisms, reconfiguration, $\Hom$-graph

\section{Introduction}\label{sect:Intro}

Given a graph $G$ and two graph homomorphisms $\alpha, \beta: G \to H$, the $H$-recoloring problem $\Recol(H)$ asks, whether $\alpha$ can be transformed into $\beta$ by
changing the image of a single vertex of $G$ at a time, maintaining a
homomorphism $G \to H$ at each step.   The $H$-recoloring problem generalizes the
graph $k$-recoloring problem, which is a prototypical example of a
combinatorial reconfiguration problem. In \cite{Cereceda:11}, Cereceda, van
den Heuvel and Johnson used  techniques resembling homotopy to show that the
$3$-recoloring problem for graphs is polynomial-time solvable in the order of the graph $G$ that is recolored.  This seems
quite surprising considering that deciding whether a given graph admits a 3-coloring is \NP-complete. The
technique was given its proper setting by Wrochna, who, making explicit use of
topology, proved the following.  A graph is {\em square-free} if it does not contain a cycle on four vertices as a subgraph, induced or otherwise.

\begin{theorem}[{\cite{Wrochna:20}}]
    \label{thm:wrochna}
    For any irreflexive square-free graph $H$, the problem $\Recol(H)$ admits a polynomial-time algorithm.
\end{theorem}

A graph homomorphism $G \to H$ is sometimes referred to as \emph{$H$-coloring} of $G$. Recall the $H$-coloring Dichotomy of Hell and Ne\v set\v ril \cite{HN:90}, which says the $H$-coloring problem admits a polynomial-time algorithm if $H$ is bipartite or contains a loop, and is \NP-complete otherwise.  The dichotomy for the constraint satisfaction problem ($\CSP$) of Bulatov and Zhuk \cite{Bulatov:17,Zhuk:20} extends this dichotomy to much more general structures.  
Similar to the $\CSP$, a dichotomy is expected to exist for $\Recol(H)$ as well.

 \begin{conjecture}
    \label{con:dichotomy}
     For any graph $H$, the problem $\Recol(H)$ is either polynomial-time solvable or \PSPACE-complete. 
 \end{conjecture}

Such a dichotomy may exist for more general relational structures $\mathbf{H}$ and in \cite{GKMP:09}, a dichotomy was given for Boolean structures (see \cite{Schwerdtfeger:14} for a small correction). However, their results do not suggest a complexity dichotomy for $\Recol(H)$ for graphs $H$.
Before the $\CSP$ dichotomy was proved, it was shown by Feder and Vardi in \cite{FV:98}, that it suffices to prove the dichotomy for directed graphs or for reflexive graphs:
the dichotomy would hold if it held in either of these  settings.  Since this sufficiency result, such structures have become staple testbed structures for $\CSP$-related problems.
In contrast to the CSP, we have no analogous sufficiency results for the $\Recol(H)$, that is, we do not know if for every relational structure $\mathbf{H}$ there is some directed graph $H$ (or reflexive graph $H$) such that $\Recol(\mathbf{H})$ is polynomially equivalent to $\Recol(H)$.

In order to make progress towards \cref{con:dichotomy},  we introduce a simple reduction,
which allows us to reduce
 reflexive instances of $\Recol(H)$ to $\Recol(\B(H))$ for some bipartite (irreflexive) graph $\B(H)$. 
The main idea of our construction is to show that a natural correspondence between homomorphisms $G
\to H$ and the homomorphisms from the vertex-edge incidence graph of $G$ to the
vertex-clique incidence graph $\B(H)$ of $H$ also holds for reconfiguration. A \emph{diamond} is the graph obtained from a $K_4$, a complete graph on four vertices, by removing exactly one edge.
For graphs $H$ without induced diamonds, the vertex-clique incidence graph $\B(H)$ is square-free. We may therefore invoke Theorem \ref{thm:wrochna} as a black box to obtain the following theorem. 
\begin{theorem}
    \label{thm:diamond-free}
    For any reflexive graph $H$ without a induced diamonds, the problem $\Recol(H)$ admits a polynomial-time algorithm for reflexive instances.
\end{theorem}

Using a strengthening of a reduction from \cite{Lee:21}, we obtain the following analogue of \cref{thm:wrochna} for reflexive graphs,

\begin{corollary} 
	\label{cor:square-free}
	For any reflexive square-free graph $H$, the problem $\Recol(H)$ admits a polynomial-time algorithm.
\end{corollary}

Notice that, in contrast to \cref{thm:wrochna}, this result is tight in the sense that if $H$ is a reflexive square then $\Recol(H)$ is \PSPACE-complete~\cite{Wrochna:20}.

Let us briefly discuss the significance of these results.
The topological approach of Wrochna that yields \cref{thm:wrochna} is robust in the sense that it has later been adapted to many different settings.  
For instance, it has been adapted to obtain polynomial-time algorithms for $\Recol(H)$ for digraphs $H$~\cite{LMS:23} as well as for reflexive graphs $H$~\cite{Lee:21}.  
While the topological approach works similarly in the various settings mentioned, it generally takes considerable work to properly tailor the topology to the setting, and so simple reductions between the settings are of considerable interest. In particular, \cref{cor:square-free} provides a strengthening of the following result from \cite{Lee:21}, which was obtained using the topological approach and represents the current state-of-the-art for recoloring homomorphisms to reflexive graphs.

\begin{corollary}[\cite{Lee:21}]
	\label{girth5}
	For any reflexive graph $H$ of girth at least $5$, the problem $\Recol(H)$ admits a polynomial-time algorithm. 
\end{corollary}

Furthermore, our reduction allows us to recover some of the known hardness results in this setting.
In \cite{Lee:20}, it was shown that $\Recol(H)$ is $\PSPACE$-complete if $H$ is
an irreflexive non-trivial (i.e., not a 4-cycle) $K_{2,3}$-free quadrangulation
of the sphere.  A similar proof was used to show that $\Recol(H)$ is
$\PSPACE$-complete for reflexive instances if $H$ is a non-trivial (i.e., not
$K_3$) $K_4$-free triangulation of the sphere.  Observing that $\B(H)$ is a
$K_{2,3}$-free irreflexive quadrangulation of the sphere when $H$ is a
non-trivial reflexive triangulation, our reduction recovers much of what is
proved for irreflexive graphs in \cite{Lee:20} from the reflexive results.  It
is not clear how much of the irreflexive case becomes redundant.  This is
discussed more in \Cref{sec:discussion}.

The reduction also has implications for the $H$-mixing problem $\Mix(H)$, which given a graph $G$ asks whether $\Recol(H)$ is true for \emph{any} two homomorphsisms $G \to H$.
In \cite{KLS:23} it was shown that $\Mix(H)$
is $\coNP$-complete, for reflexive instances, if $H$ is a reflexive graph of
girth at least $4$. Such graphs $H$ do not contain induced diamonds, so by our
reduction, the problem $\Mix(\B(H))$ is also $\coNP$-complete.  This gives a
wealth of new examples of bipartite graphs for which $\Mix(H)$ is
$\coNP$-complete. 

Finally, given that our reduction shows how to obtain a polynomial-time algorithm for $\Recol(H)$ for reflexive graphs $H$ from a polynomial-time algorithm for $\Recol(H)$ for irreflexive graphs $H$, we may ask whether there is a reduction in the other direction as well. This would show the equivalence of the reflexive and irreflexive settings. We make partial progress in this direction, providing a reduction from the irreflexive to the reflexive setting for \emph{bipartite} instances.
We denote by $H \times K_2$ the categorical product of $H$ with the graph $K_2$ (see \cref{sec:prelim} for a definition) and by $G^\circ$ the graph obtained from a graph $G$ by adding a loop to each vertex.

\begin{proposition} 
    \label{prop:reflexive}
    Let $H$ be an irreflexive square-free graph. Then $\Recol(H)$ restricted to bipartite instances can be polynomially reduced to $\Recol((H \times K_2)^\circ)$.
\end{proposition}

\section{Preliminaries}
\label{sec:prelim}

All graphs we consider are undirected, finite, and without parallel edges, but
possibly with loops. We further consider graphs to be connected and non-empty (having at least one edge) as the reductions to this state for questions of recoloring are easy and standard.  As mentioned above, a graph is {\em reflexive} graph if every vertex has a loop. A loopless graph is called {\em irreflexive}.  For a graph $G$, we
denote by $G^\circ$ the graph we get from $G$ by adding loops to every vertex. For graphs  $G, K$, we say that $G$ is $K$-free, if $G$ does not contain $K$ as subgraph. We denote by $K_\ell$ the complete graphs on $\ell$ vertices. A \emph{diamond} is the graph obtained from $K_4$ by removing an edge. We denote by $C_\ell$ a cycle on $\ell$ vertices and call the graph $C_3 = K_3$ a \emph{triangle} and $C_4$ a \emph{square}. The (categorical) product $A \times B$ of two graphs $A$ and $B$ has vertex set $V(A\times B) = V(A) \times V(B)$, and two vertices $(a,b)$ and $(a',b')$ are adjacent if $a \sim a'$ and $b \sim b'$.
A graph homomorphism
$\alpha : G \to H$ is a mapping $V(G) \to V(H)$, such that $\alpha(u)\alpha(v)
\in E(H)$, whenever $uv \in E$.  As a homomorphism to $H$ is called an
$H$-coloring, the well known $H$-coloring problem $\Hom(H)$ asks if a given
graph $G$ admits a homomorphism to $H$.  

The main reconfiguration versions of $\Hom(H)$ consider the connectivity
of the $\Hom$-graph: Given graphs  $G$ and $H$, the $\Hom$-graph $\Hom(G,H)$ is
the graph whose vertices are homomorphisms $G \to H$ and for $\alpha, \beta
\colon G \to H$, there is an edge $\alpha\beta$ in $\Hom(G,H)$, if for any edge
$uv$ of  $G$, there is an edge $\alpha(u)\beta(v)$ in $H$. In particular,
$\Hom(G,H)$ is reflexive. The {\em $H$-recoloring graph} $\rCol(G,H)$, sometimes denoted $\Hom_1(G,H)$, is the subgraph of
$\Hom(G,H)$ consisting only of the edges between maps that differ on at most one vertex.   

Fix a graph $H$.
Given a graph $G$ and two homomorphisms $\alpha, \beta : G \to H$, the {\em $H$-recoloring problem} $\Recol(H)$ asks if there is a path between $\alpha$ and $\beta$ in $\rCol(G,H)$.
The {\em $H$-reconfiguration problem} $\Recon(H)$ asks if there is a path
between them in $\Hom(G,H)$. While these are distinct problems for digraphs,
they are well known to be equivalent for graphs. 
Though an instance of $\Recol(H)$ is $(G, \alpha, \beta)$, it is common to
refer to an instance just as the graph $G$ if the two homomorphisms are
irrelevant. So an instance $(G, \alpha, \beta)$ is called, say, bipartite
(resp., reflexive), if $G$ is bipartite (resp., reflexive).

While it is easy to show that a homomorphism $H \to H'$ induces a homomorphism $\Hom(G,H) \to \Hom(G,H')$ of $\Hom$-graphs, in \cite{BW:00}, Brightwell and Winkler observed a stronger relation for dismantling retractions $H \to H'$.  Given a graph $H$, a vertex $b$ of $H$ \emph{folds} into the graph $H' := H \setminus \lbrace b \rbrace \subset H$ if there is another vertex $a$ of $H$ such that all neighbors of $b$ are also neighbors of $a$. The homomorphism $H \to H'$ that fixes $H'$ and maps $b$ to $a$ is a dismantling retraction. It was shown in \cite{BW:00} that a dismantling retraction $H \to H'$ not only induces a homomorphism $\Hom(G,H) \to \Hom(G,H')$, 
but that any path between vertices of $\Hom(G,H')$ lifts to a path between them in $\Hom(G,H)$. So the problems $\Recol(H)$ and $\Recol(H')$ are polynomially equivalent, meaning, in particular, that if one of them is polynomial time solvable, or $\PSPACE$-complete, then so is the other.   
A graph $H$ {\em dismantles} to a subgraph $H^*$ if one can reduce $H$ to $H^*$ by a series of dismantling retractions.  By induction one gets the following basic fact.

\begin{fact}\label{dismant}
 If a graph $H$ dismantles to a subgraph $H'$ then the problems $\Recol(H)$ and $\Recol(H')$ are polynomially equivalent. 
\end{fact}

From this we know, for example, that $\Recol(H)$ is not only polynomial time solvable, but trivial, meaning that all instances are YES instances, if $H$ dismantles to a loop, and if $H$ dismantles to an edge, then $\Recol(H)$ is just the
problem $\Recol(K_2)$ of deciding if two bipartitions of a graph are the same bipartition.

\section{Known reductions}\label{sect:otherreductions} 
\label{sec:known}

The ideal of a reduction for recoloring is one that says that for every graph $H$ there is a graph $H'$, in some restricted class of graphs, such that $\Recol(H)$ is polynomially equivalent to $\Recol(H')$.  With such a result, one gets a complexity dichotomy for recoloring by proving it on the restricted class. We mentioned such reductions for $\CSP$ in \cref{sect:Intro}: Feder and Vardi showed that for every relational structure $\mathbf{H}$ there is a reflexive graph $H'$,
 such that $\CSP(\mathbf{H})$ is polynomially equivalent to $\CSP(H')$.
 This reduction does not extend to recoloring. 


Another reduction for $\CSP$, which is extremely useful in the the proof of the $H$-coloring dichotomy, but less useful in dichotomies for, say,  list versions of $H$-coloring, is that it is enough consider {\em core} structures --- structures whose only endomorphisms are automorphisms.  The reduction to cores does not extend to recoloring, but we see by \cref{dismant} that we can recover part of it and restrict our attention to so-called {\em stiff} graphs, that is, graphs that admit no folds.

 We are aware of only one other relatively simple construction that allows one to reduce $\Recol(H)$ to $\Recol(H')$ for some graph $H' \neq H$. It uses graph products, so we call it the product reduction; it appears in Wrochna's Master's thesis~\cite{Wrochna:14}. 
Wrochna observed that $\Recol(H)$ for $R$-colorable
instances admits a polynomial reduction to $\Recol(H \times R)$. The most
useful case of this is when $R$ is the edge $K_2$, showing that bipartite
instances of $\Recol(H)$ can be polynomially reduced to $\Recol(H \times K_2)$.
From this Wrochna concluded that $H$-Recoloring is \PSPACE-complete when $H$ is
a reflexive 4-cycle~\cite{Wrochna:20}. This follows because $K_4 \times K_2 = C_4^\circ \times
K_2$, where $C_4^\circ$ is the reflexive $4$-cycle, and because $\Recol(K_4)$
is known to be $\PSPACE$-complete even for bipartite instances~\cite{BC:09}.
This reduction would be much more applicable in showing that $\Recol(H)$ is
polynomial-time solvable if it allowed us to reduce not just the bipartite instances,
but all instances of $\Recol(H)$ to $\Recol(H \times K_2)$.   

We note that the product construction has one more obvious application. Taking
$R = H$ it tells us that if $\Recol(H^n)$ is \PSPACE-complete for some $n$ then
$\Recol(H^\ell)$ is \PSPACE-complete for any $\ell \geq n$.  
So our proposed construction reduces reflexive instances of $\Recol(H)$ for reflexive $H$ to an irreflexive recoloring problem, and the product construction reduces bipartite instances of $\Recol(H)$ for any graph $H$ to an irreflexive recoloring problem.

\section{Main reduction: reflexive to irreflexive recoloring} 
\label{sec:reduction}

In this section we prove \cref{thm:diamond-free} by providing a reduction to Theorem
\ref{thm:wrochna}. We first provide some intuition of the construction along with an example in \cref{sec:intuition} and then prove \cref{thm:diamond-free} in \cref{sec:proof}.

\subsection{Intuition} 
\label{sec:intuition}

Let $G$ and $H$ be reflexive graphs and consider any homomorphism $\alpha : G \to H$.
A vertex $v$ of $G$ mapped to a vertex $\alpha(v)$ of $H$, having a loop, can
only reconfigure under an edge of $\Hom(G,H)$ to a neighbor of $\alpha(v)$.  
Moreover, since homomorphisms map edges to edges, it can only move within a clique that contains the images of the neighbors of $v$ in $G$. 
This suggests the following definitions.  

For a reflexive graph $H$ let the \emph{vertex-clique incidence graph} $\cK(H)$
be the bipartite graph with vertex set $V \cup M$, where $V = V(H)$ and $M =
M(H)$ is the set of maximal cliques of $H$, and in which $v \in V$ is adjacent
to $K \in M$ if $v$ is in $K$. For an example of a graph and its vertex-clique
incidence graph, see Figure~\ref{fig:K(H)}. 
For a reflexive graph $G$ let the \emph{vertex-edge incidence graph} $\cE(G)$
of $G$ be the bipartite graph with vertex set $V \cup E$, where $V = V(G)$ and
$E = E(G)$, and in which $v \in V$ is adjacent to $e \in E$ if $v$ is in $e$.  
Since distinct maximal cliques that share an edge must have non-adjacent vertices that are adjacent to both endpoints of the edge, we observe the following.
\begin{fact}\label{diamond}
  Let $H$ be a graph without induced diamond. Then each edge of $H$ is contained in a 
  unique maximal clique. 
\end{fact}

\begin{figure}[t]
\begin{center}
\begin{tikzpicture}[node distance={2.5cm}, thick, main/.style = {draw, circle,fill=white}] 
\begin{scope}
	\node[main,label = 90:0] (0) at (0,0){};
 	\node[main,label = 120:1] (1) at (-2,1){};
	\node[main,label = -120:2] (2) at (-2,-1){};
	\node[main,label = -60:3] (3) at (2,-1){};
	\node[main,label = 60:4] (4) at (2,1){};
	\node[main,label = -120:5] (5) at (-1,-3){};
	\node[main,label = -60:6] (6) at (1,-3){};

	\draw[color=blue,thick] (0) -- (1) -- (2) -- (0) -- (3) -- (4) -- (0);
	\draw[color=red,thick] (2) -- (5) -- (6) -- (3); 
	\draw [blue!30,thick] (0.center) -- (1.center) -- (2.center) -- cycle;
	\draw [blue!30,thick] (0.center) -- (3.center) -- (4.center) -- cycle;
	
	\node[main] (0) at (0,0){};
	\node[main] (1) at (-2,-1){};
	\node[main] (2) at (-2,1){};
	\node[main] (3) at (2,-1){};
	\node[main] (4) at (2,1){};
	\node[main] (5) at (-1,-3){};
	\node[main] (6) at (1,-3){};

\end{scope}

\begin{scope}[shift = {(7,-1)}]

	\node[main,label = 180:0] (0) at (-1,3){};
	\node[main,label = 180:1] (1) at (-1,2){};
	\node[main,label = 180:2] (2) at (-1,1){};
	\node[main,label = 180:3] (3) at (-1,0){};
	\node[main,label = 180:4] (4) at (-1,-1){};
	\node[main,label = 180:5] (5) at (-1,-2){};
	\node[main,label = 180:6] (6) at (-1,-3){};
	
	\node[main,fill=blue,label = 0:012] (T1) at (1,2){};
	\node[main,fill=blue,label = 0:034] (T2) at (1,1){};
	
	\node[main,fill=red,label = 0:25] (E1) at (1,0){};
	\node[main,fill=red,label = 0:56] (E2) at (1,-1){};
	\node[main,fill=red,label = 0:36] (E3) at (1,-2){};
	
	\draw (T1) -- (0);
	\draw (T1) -- (1);
	\draw (T1) -- (2);
	\draw (T2) -- (0);
	\draw (T2) -- (3);
	\draw (T2) -- (4);
	\draw (E1) -- (2);
	\draw (E1) -- (5);
	\draw (E2) -- (5);
	\draw (E2) -- (6);
	\draw (E3) -- (6);
	\draw (E3) -- (3);
	
\end{scope}
\end{tikzpicture}
\end{center}

\caption{On the left, a graph $H$ whose maximal cliques are the two triangles in blue and three edges in red. On the right, the vertex-clique incidence graph $\B(H)$.} \label{fig:K(H)}
\end{figure}
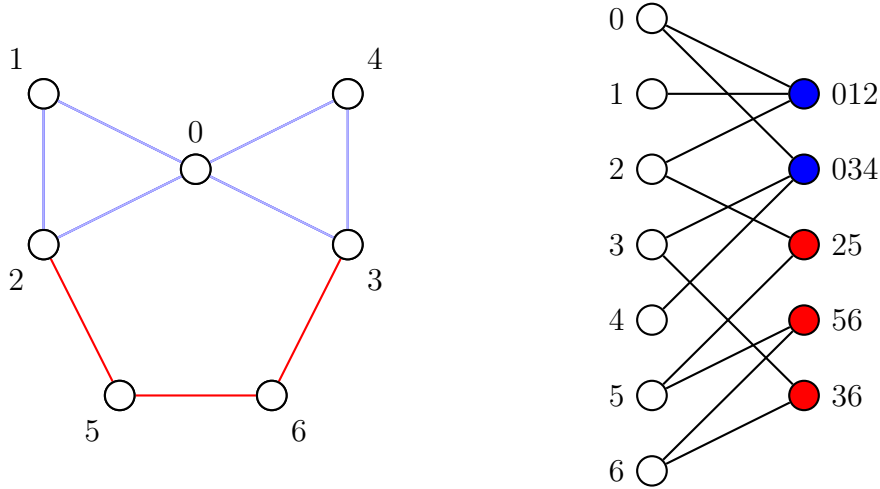

The main technical step is showing that  we can solve $\Recol(H)$ by solving $\Recol(\B(H))$ (\cref{prop:eqtowro}).
To do this, we define a map $K$ from $\rCol(G,H)$ to $\rCol(\mathcal{E}(G),\B(H))$, 
that takes a map $\alpha$ to a map $\alpha^K$, and show that there is a path between 
$\alpha$ and $\beta$ of in $\rCol(G,H)$ if and only if there is a path between 
$\alpha^K$ and $\beta^K$ in $\rCol(\mathcal{E}(G),\B(H))$.

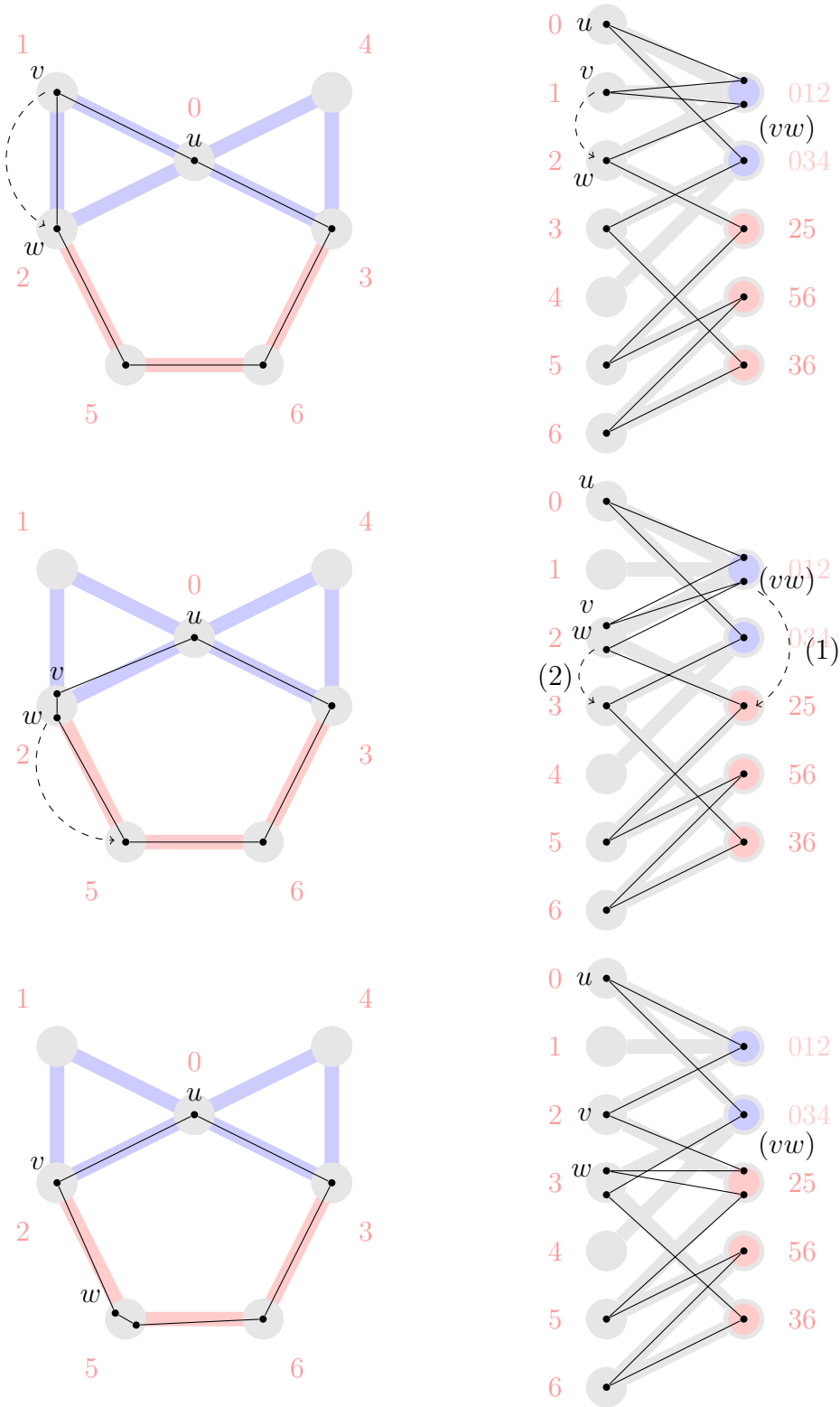
\begin{figure}[tbh]

\begin{tikzpicture}
\begin{scope}
	\dsevengraph;
	\node[G, label = 90:$u$] (0) at ($(h0)$){};
	\node[G, label = 120:$v$] (1) at ($(ha1)$){};
	\node[G, label = -120:$w$] (2) at ($(ha2)$){};
	\node[G] (3) at ($(hc1)$){};
	\node[G] (4) at ($(hc2)$){};
	\node[G] (5) at ($(hb1)$){};
	
	\draw[Ge,->,dashed] ($(ha1) - (0:5pt)$) arc (120:240:1.13);
	\draw[Ge] (0) -- (1) -- (2) -- (3) -- (4) -- (5) -- (0);

\end{scope}

\begin{scope}[shift = {(7,-1)}]
	\dBsevengraph;
	\node[G, label = 180:$u$] (0) at ($(h0)$){};
	\node[G] (K0) at ($(hT1) + (90:5pt)$) {};
	\node[G, label = 120:$v$] (1) at ($(h1)$){};
	\node[G, label = -10:$(vw)$] (K1) at ($(hT1) - (90:5pt)$) {};
	\node[G, label = -120:$w$] (2) at ($(h2)$){};
	\node[G] (K2) at ($(hE1)$){};
	\node[G] (3) at ($(h5)$){};
	\node[G] (K3) at ($(hE2)$){};
	\node[G] (4) at ($(h6)$){};
	\node[G] (K4) at ($(hE3)$){};
	\node[G] (5) at ($(h3)$){};
	\node[G] (K5) at ($(hT2)$){};
	
	\draw[Ge,->,dashed] ($(h1) - (0:5pt)$) arc (120:240:0.55);
	\draw[Ge] (0) -- (K0) -- (1) -- (K1) -- (2) -- (K2) -- (3) -- (K3) -- (4) -- (K4) -- (5) -- (K5) -- (0);

\end{scope}

\begin{scope}[shift = {(0,-7)}]
	\dsevengraph;
	\node[G, label = 90:$u$] (0) at ($(h0)$){};
	\node[G, label = 90:$v$] (1) at ($(ha2) +(90:5pt)$){};
	\node[G, label = 180:$w$] (2) at ($(ha2) -(90:5pt)$){};
	\node[G] (3) at ($(hc1)$){};
	\node[G] (4) at ($(hc2)$){};
	\node[G] (5) at ($(hb1)$){};
	
	\draw[Ge,->,dashed] ($(2) - (30:5pt)$) arc (150:270:1.14);;
	\draw[Ge] (0) -- (1) -- (2) -- (3) -- (4) -- (5) -- (0);

\end{scope}

\begin{scope}[shift = {(7,-8)}]
	\dBsevengraph;
	\node[G, label = 120:$u$] (0) at ($(h0)$){};
	\node[G] (K0) at ($(hT1) + (90:5pt)$) {};
	\node[G, label = 120:$v$] (1) at ($(h2) + (90:5pt)$){};
	\node[G, label = 0:$(vw)$] (K1) at ($(hT1) - (90:5pt)$) {};
	\node[G, label = 170:$w$] (2) at ($(h2) - (90:5pt)$){};
	\node[G] (K2) at ($(hE1)$){};
	\node[G] (3) at ($(h5)$){};
	\node[G] (K3) at ($(hE2)$){};
	\node[G] (4) at ($(h6)$){};
	\node[G] (K4) at ($(hE3)$){};
	\node[G] (5) at ($(h3)$){};
	\node[G] (K5) at ($(hT2)$){};
	
	\node at ($(hT2) + (-10:33pt)$) {$(1)$};
	\node at ($(h1) + (-115:50pt)$) {$(2)$};
	
	\draw[Ge,->,dashed] ($(2) - (0:5pt)$) arc (120:240:0.45);
	\draw[Ge,<-,dashed] ($(K2) + (0:5pt)$) arc (-60:60:0.98);
	\draw[Ge] (0) -- (K0) -- (1) -- (K1) -- (2) -- (K2) -- (3) -- (K3) -- (4) -- (K4) -- (5) -- (K5) -- (0);

\end{scope}

\begin{scope}[shift = {(0,-14)}]
	\dsevengraph;
	\node[G, label = 90:$u$] (0) at ($(h0)$){};
	\node[G, label = 120:$v$] (1) at ($(ha2)$){};
	\node[G, label = 170:$w$] (2) at ($(hc1) +(150:5pt)$){};
	\node[G] (3) at ($(hc1) -(150:5pt)$){};
	\node[G] (4) at ($(hc2)$){};
	\node[G] (5) at ($(hb1)$){};
	
	\draw[Ge] (0) -- (1) -- (2) -- (3) -- (4) -- (5) -- (0);

\end{scope}

\begin{scope}[shift = {(7,-15)}]
	\dBsevengraph;
	\node[G, label = 180:$u$] (0) at ($(h0)$){};
	\node[G] (K0) at ($(hT1)$) {};
	\node[G, label = 180:$v$] (1) at ($(h2)$){};
	\node[G, label = 10:$(vw)$] (K1) at ($(hE1) + (90:5pt)$) {};
	\node[G, label = 180:$w$] (2) at ($(h3) + (90:5pt)$){};
	\node[G] (K2) at ($(hE1) - (90:5pt)$){};
	\node[G] (3) at ($(h5)$){};
	\node[G] (K3) at ($(hE2)$){};
	\node[G] (4) at ($(h6)$){};
	\node[G] (K4) at ($(hE3)$){};
	\node[G] (5) at ($(h3) - (90:5pt)$){};
	\node[G] (K5) at ($(hT2)$){};
	
	\draw[Ge] (0) -- (K0) -- (1) -- (K1) -- (2) -- (K2) -- (3) -- (K3) -- (4) -- (K4) -- (5) -- (K5) -- (0);

\end{scope}

\end{tikzpicture}

\caption{On the left, two reconfiguration moves starting from the homomorphisms $G \to H$ of Figure~\ref{fig:K(H)}. On the right, the associated reconfiguration moves of homomorphisms $\mathcal{E}(G) \to \B(H)$. Notice that before $w$ moves, it is required to change the position of the edge $(vw)$ in $\B(H)$ first (since the edges $(vu)$ and $(vw)$ remain on same blue triangle after $v$ moves, this requires no change in $\B(H)$).  } \label{fig:recolB(H)}
\end{figure}

Assume that $H$ contains no induced diamond and consider a homomorphism $\alpha: G \to H$.  Let $\alpha^K$ be the map from $\cE(G)$ to $\B(H)$ that restricts on $V(G) \subset V(\cE(G))$ to $\alpha$, and maps a edge $e$ of $G$ in $E(G) \subset V(\cE(G))$ to the unique maximal clique of $H$ that contains $\alpha(e)$.  We will show that this is a homomorphism of $\cE(G)$ to $\B(H)$. Figure \ref{fig:recolB(H)} shows several maps $\alpha$ and the associated maps $\alpha^K$.  This construction also commutes with reconfiguration. 
Referring to left side of the figure, consider the vertex $u$ mapped to the colour $2$.  We can not move it to the to the color  $3$, as it has the neighbor $v$ that is mapped to the color $1$.  This can be seen in terms of cliques on the right side of the figure. The vertex $v$ is mapped to a maximal clique of $H$ that does not contain $3$.
We would have to move $v$ to the color $0$ first, and then move $u$ to $3$.
Looking at the clique graph, $\B(H)$, we see exactly this argument. The vertex 
$u$ cannot move to $3$ because its neighbor (the edge $uv$ now) is mapped to the
clique $012$.  Moving $v$ to $0$ corresponds to moving $uv$ to the
clique $034$, which then allows $u$ to move to $3$.

There is one more detail of the construction. According to \cref{diamond}, any edge $e = uv$ of $H$ is in a unique maximal clique $\Cl_{uv}$ of $H$.  Our intuition for a map from $\rCol(G,H)$ to $\rCol(\mathcal{E}(G),\B(H))$ was to take $\alpha$ to the map taking $uv$ to $\Cl_{uv}$. But $H$ is reflexive, so we need to know where to take $uv$ if $\alpha$ maps it to a loop.  We replace $\B(H)$ with the graph $\B'(H)$ that we get from $\B(H)$ by adding, for each $a \in V(H)$, the `false maximal clique' $\Cl_{aa}:= \lbrace a \rbrace$ adjacent only to the vertex $a$. 
\clearpage

For graphs $H$ without
induced diamond, $\B(H)$ and $\B'(H)$, are square-free, so we may invoke
\cref{thm:wrochna}. 
Though not evident in the example, the construction will not work for
irreflexive instances. The following result allows us to restrict our attention
to reflexive instances $G$ if $H$ is square-free. It was stated in \cite{Lee:21} for $H$ being a reflexive graph of girth at least $5$, but the proof only used the assumption that  
$H$ is square-free, so we state it this way.

\begin{lemma}[{see \cite{Lee:21}}]
    \label{lem:equirecol} 
    Let $G$ be a graph and let $H$ be a reflexive square-free graph. If two homomorphisms $\alpha$ and $\beta$ are adjacent in $\rCol(G,H)$, then there is a path of length at most two between them in $\rCol(G^\circ, H)$. 
\end{lemma}

So finally, we may get rid of the restriction to reflexive
instances $G$ by invoking \cref{lem:equirecol} and obtain \cref{cor:square-free}.

\subsection{Proof of \cref{thm:diamond-free}}
\label{sec:proof}

In the following, fix a reflexive graph $H$ that contains no induced diamond. Furthermore, let $G$ be a reflexive graph, that is, an instance of $\Recol(H)$.
Recall that we denote by $\mathcal{E}(G)$ the vertex-edge incidence graph of $G$ and by $\B'(H)$ the vertex-clique incidence graph of $H$ with a private clique $K_a := \{a\}$ adjacent only to the vertex $a \in V(\B(H))$ for each $a \in V(H)$. Let $\alpha : G \to H$ be a homomorphism. We first define the associated homomorphism $\alpha^K : \mathcal{E}(G) \to \B'(H)$.

\begin{definition}
Let $\alpha^K \colon \mathcal{E}(G) \to \B'(H)$ be the natural graph homomorphism  defined
on $V(G)$ by $\alpha^K(v) = \alpha(v)$ and on $E(G)$ by
\[ \alpha^K((uv)) = \Cl_{\alpha(u)\alpha(v)}\enspace. \] 
\end{definition}

We can easily check that $\alpha^K$ is indeed a graph homomorphism: for any vertex $u \in V(G)$ and any edge $(uv) \in E(G)$ containing $u$ (a neighbor of $u$ in $\mathcal{E}(G)$), we have that $\alpha^K(e)$ is a clique containing $\alpha(u)$, so $\alpha^K(u)\alpha^K(e)$ is an edge of $\B'(H)$. Conversely, one can easily check for any homomorphism $\alpha \colon \mathcal{E}(G) \to \B'(H)$, that the restriction $\alpha^\circ$ of $\alpha$ to $V(G)$ is a homomorphism $\alpha \colon G \to H$. Indeed, for adjacent vertices $u$ and $v$ of $G$, we have that $\alpha((uv))$ is a clique containing both $\alpha(u) = \alpha^\circ(u)$ and $\alpha(v) = \alpha^\circ(v)$, so these are neighbors in~$H$.

\begin{proposition} 
    \label{prop:eqtowro} 
    
    Let $G$ and $H$ be reflexive graphs and let $\alpha,\beta \colon G \to H$. Then there is a path from $\alpha$ to $\beta$ in $\rCol(G,H)$ if and only if there is a path from $\alpha^K$ to $\beta^K$ in $\rCol(\mathcal{E}(G),\B(H))$.
\end{proposition}
\begin{proof} 
   We prove the result with $\B'(H)$ in place of $\B(H)$, it is clear that $\B'(H)$ dismantles to $\B(H)$, so the result holds as stated, by \cref{dismant}. 

   On the one hand, assume there is a path $\alpha=\sigma_0, \sigma_1, \ldots, \sigma_\ell = \beta$ from $\alpha$ to $\beta$ in $\rCol(\mathcal{E}(G),\B'(H))$, and let $\sigma_i\sigma_{i+1}$ be an edge of this path.  If the vertex on which $\sigma_i$ and $\sigma_{i+1}$ differ is in $E(G) \subset V(\cE(G))$ then their restrictions to $V(G)$ are the same map, so are adjacent as $\rCol(G,H)$ is reflexive.  If the vertex $u$ on which they differ is in $V(G)$, then their restrictions $\sigma_{i}^\circ$ and $\sigma_{i+1}^\circ$ also differ on only this vertex.  Since these are homomorphisms $G \to H$, we have that for any neighbor $v$ of $u$ in $G$, $\sigma_i(u)\sigma_{i+1}(v) = \sigma_i(u)\sigma_i(v) \in E(H)$, so $\sigma_i^\circ \sim \sigma_{i+1}^\circ$.

On the other hand,  assume that $\alpha$ and $\beta$ are neighbors in $\rCol(G,H)$. Let $u \in V(G)$ be the vertex on which $\alpha$ and $\beta$ differ. Let $v$ be any neighbor of $u$ in $G$ and let $h := \alpha(v) = \beta(v)$.   As $h = \beta(v)$ is in the edge  $\beta(u)\beta(v)$, it is
in the clique $\beta^K((uv))$ 
which means that $\beta^K(v)\beta^K((uv))$ is an edge of $\B'(H)$. Hence, $(uv)$ can be recolored from $\alpha^K((uv))$ to $\beta^K((uv))$. 
\end{proof}

By~\cite[Theorem 8.2]{Wrochna:20} (essentially \cref{thm:wrochna}), we obtain the following result, which immediately implies \cref{thm:diamond-free}.

\begin{theorem}\label{thm:reduction} 
    Let $H$ be a reflexive symmetric graph without induced diamond. Then $\Recol(H)$ can be solved in time $O(|E(\mathcal{E}(G))| \cdot ( |V(\mathcal{E}(G))| + |E(\B(H))|) = O(|E(G)|\cdot |V(G)| \cdot ( |(E(G)| + |E(H)|^2)$ for any reflexive instance $G$. 
\end{theorem}
Applying Lemma~\ref{lem:equirecol}, we finally obtain \cref{cor:square-free}.

\section{Reduction from irreflexive to reflexive recoloring}
\label{sec:inverse}

In this section we prove \cref{prop:reflexive} by giving a partial inverse of the reduction from \cref{sec:reduction}.  
That is, we show that when restricted to bipartite instances, Corollary~\ref{cor:square-free} implies Theorem~\ref{thm:wrochna}. 
Notice that \cref{prop:reflexive} cannot be used in order to show hardness of $\Recol(H)$, since the assumption that $H$ is square-free means that there is a polynomial-time algorithm by \cref{thm:wrochna}. We wonder whether there is a construction that can reduce $\Recol(H)$ for any $H$ to a recoloring problem with a reflexive target.
 
Let $G$ and $H$ be irreflexive graphs, such that $G$ is bipartite and $H$ is square-free. Let $V(G) = G_0 \sqcup G_1$ be a bipartition of $G$ and $V(H \times K_2) = H_0 \sqcup H_1$ be a bipartition of $H \times K_2$. Observe that when $G$ is connected and not an isolated vertex, any homomorphism $G \to H \times K_2$ must either map the vertices of $G_0$ to $H_0$ and those of $G_1$ to $H_1$, or the converse; this is preserved by reconfiguration moves. Recall that we denote by $(H \times K_2)^\circ$ the graph $H \times K_2$ with a loop added on each vertex.

\begin{lemma} 
    \label{lem:eqbipartite}
Assume that $H$ is square-free. Let $\alpha, \beta \colon G \to H \times K_2$ be homomorphisms that map vertices of $G_0$ to the same bipartite component of $H \times K_2$. There is a path from $\alpha$ to $\beta$ in $\Col(G, H \times K_2)$ if and only if there is a path from $\alpha$ to $\beta$ in $\Col(G, (H \times K_2)^\circ)$.
\end{lemma}

\begin{proof} Since $H \times K_2$ is a subgraph of $(H \times K_2)^\circ$, we have that $\Col(G, H \times K_2)$ is a subgraph of $\Col(G, (H \times K_2)^\circ)$. Therefore, any path in $\Col(G, H \times K_2)$ is also a path in $\Col(G, (H \times K_2)^\circ)$.
On the other hand, assume that there a path from $\alpha$ to $\beta$ in $\Col(G, (H \times K_2)^\circ)$. Without loss of generality, assume that $\alpha$ and $\beta$ map the vertices of $G_0$ to $H_0$ and those of $G_1$ to $H_1$. If a vertex $u \in G_0$ (resp., $u \in G_1$) is mapped by a homomorphism $\alpha : G \to (H \times K_2)^\circ$ to $H_1$ (resp., $H_0$), we say that it is on the \emph{wrong} side; otherwise it is on the \emph{right} side. We assumed above that for $\alpha$ and $\beta$, all vertices are on the right side. 

We view the path from $\alpha$ to $\beta$ in $\Col(G, (H \times K_2)^\circ)$ as a sequence $S = (C_1, \ldots, C_\ell)$ of color changes. If after any color change $C_i$, we have a homomorphism $\alpha$ that maps a vertex of $G$ to the wrong side, then $\alpha$, even if it is a homomorphism $G \to H \times K_2$, is not reachable from $\alpha$ or $\beta$ in $\Col(G, H \times K_2)$. So in this case we modify $S$ such that the vertices of $G$ never get recolored to the wrong side.
For each color change $C_i$ in $S$, we thus associate a color change $C'_i$ defined as follow:
\begin{itemize}
    \item If $C_i$ moves a vertex $u$ of $G$ to the right side, then $C'_i$ moves $u$ to the same color.
    \item If $C_i$ moves a vertex $u$ from the right side to the wrong side, then $C'_i$ does nothing.
    \item If $C_i$ moves a vertex $u$ from a color $a$ on the wrong side to a color $c$ on the wrong side, then $C'_i$ moves $u$ to the unique common neighbor $b$ of $a$ and $c$ in $H \times K_2$ ($b$ is unique since $H \times K_2$ is square-free).
\end{itemize}
The modified sequence of color changes is $S' = (C'_1, \ldots, C'_\ell)$. Notice that changing the color a vertex gets recolored to also changes the initial color of that vertex the next time it gets recolored. For this reason, we did not specify the initial color for the $C'_i$. It is easy to see that after a sequence of color changes $(C'_1, \ldots, C'_i)$ (for $i \leq \ell$), all vertices of $G$ are mapped to the right side. However, it is less clear that the vertex mapping is a homomorphism $G \to H \times K_2$. We show next by induction on $i$.
Let $i \leq \ell$ and assume that $\alpha$, modified by the sequence of color changes $(C'_1, \ldots, C'_i)$, defines a homomorphism $G \to H \times K_2$. Let $u$ be the vertex of $G$ whose color is changed by $C_{i+1}$ and $v$ be a neighbor of $v$.
\begin{itemize}
    \item If $u$ is recolored from a color $a$ on the right side to a color $c$ on right side by $C_{i+1}$, then $C'_{i+1}$ is exactly the same as $C_{i+1}$. During $C_{i+1}$ the color of $v$ must be the unique neighbor $b$ of $a$ and $b$ in $H \times K_2$; since $b$ is on the right side for $v$, it is also the color of $v$ during $C'_{i+1}$ so we have a homomorphism afterwards.
    \item If $u$ is recolored from the right side to the wrong side by $C_{i+1}$, then $C'_{i+1}$ does nothing so we have a homomorphism afterwards.
    \item If $u$ is recolored from a color $b$ on the wrong side to a color $c$ on the right side by $C_{i+1}$, then $v$ must have color either $b$ or $c$ during $C_{i+1}$. If $v$ has color $b$, which is on the right side for $v$, then this also happens during $C'_{i+1}$. If $v$ has color $c$, then the last color change of $v$ in $S$ before $C_{i+1}$ was a move from a neighbor $a$ of $b$ (the color of $u$) or from $b$ itself. In any case, the corresponding move in $S'$ must have left $v$ with color $b$. Eventually, the edge $(uv)$ of $G$ gets correctly mapped to the edge $(bc)$ of $H \times K_2$.
    \item If $u$ is recolored from a color $b$ on the wrong side to a color $d$ on the wrong side by $C_{i+1}$, then $v$ must have the unique common neighbor $c$ of $b$ and $d$ during $C_{i+1}$; notice that $c$ is on the wrong side for $v$. The last color change of $v$ in $S$ before $C_{i+1}$ was a move from a neighbor $a$ of $b$ (the color of $u$) or from $b$ itself. In any case, the corresponding move in $S'$ must have left $v$ with color $b$. Eventually, the edge $(uv)$ of $G$ gets correctly mapped to the edge $(cb)$ of $H \times K_2$.
\end{itemize}

\end{proof}

Wrochna proved in~\cite{Wrochna:14} that bipartite instances of $\Recol(H)$ can be polynomially reduced to $\Recol(H \times K_2)$, so \cref{lem:eqbipartite} yields \cref{prop:reflexive}.


\section{Discussion}
\label{sec:discussion}

We showed by a simple reduction that $\Recol(H)$ admits a polynomial-time algorithm whenever $H$ is reflexive and square-free (\cref{cor:square-free}). This is a strengthening of~\cite[Corollary 1.2]{Lee:21} (see \cref{girth5}).
We note that $\B(K_4^\circ)$ is the $3$-cube, so this is one example where  $\Recol(\B(H))$ is $\PSPACE$-complete while $\Recol(H)$ is not. Furthermore, we showed that a polynomial reduction in the other direction exists for bipartite instances (\cref{prop:reflexive}).

In order to make progress towards \cref{con:dichotomy}, a goal would be a construction similar to those from \cite{FV:98} that yields for any graph $H$, an irreflexive graph $H'$, such that $\Recol(H)$ is polynomially equivalent to $\Recol(H')$.  The same would be interesting with reflexive graphs $H'$.   
While this is the ideal, it seems difficult, and so a one-way construction, which gives a polynomial reduction rather than a polynomial equivalence, is the next best thing.

As mentioned in the introduction, it was shown in \cite{Lee:20} that $\Recol(H)$ is $\PSPACE$-complete if $H$ is an irreflexive non-trivial $K_{2,3}$-free quadrangulation of the sphere, and then the same was showed when $H$ is a reflexive non-trivial $K_4$-free triangulation of the sphere.  We claimed that when $H$ is a non-trivial reflexive triangulation of the sphere, then $\B(H)$ is a $K_{2,3}$-free quadrangulation of the sphere, and so many of their reflexive results follow from their irreflexive result by \cref{prop:eqtowro}.
To see this, one simply has to observe that the maximal cliques in $H$ are its triangles (recall that `non-trivial' implied that $H$ was is not $K^\circ_4$). Embedding them in the middle of the triangle they represent, the edges of $\B(H)$ go to the vertices on the triangle they are on.  Every edge of $H$ is `replaced' by the quadrangle consisting of its endpoints and the vertices representing the two triangles that it bounds.  It is not hard to see that the only way this gives a $K_{2,3}$ is if $H$ is a triangle; this is also excluded as a 'trivial' case.  

We wonder to what extent the irreflexive results of \cite{Lee:20} are made redundant by \Cref{prop:eqtowro}.  Is every $K_{2,3}$ quadrangulation $Q$ of a sphere equal to $\B(H)$ for some triangulation?
Actually, the results in \cite{Lee:20} are a bit stronger than this. It would be enough too show that every $K_{2,3}$-free quadrangulation $Q$ of a sphere is equal to $\B(H)$ where $H$ is a stiff two-connected planar graph with some vertex of degree at least $4$ all of whose faces are triangles.

\bibliographystyle{plain}
\bibliography{mainArxiv}

\end{document}